\documentclass[conference]{IEEEtran}
\usepackage{subfigure, epsfig, color}
\usepackage{amsmath}
\usepackage{amsthm}
\usepackage{amssymb}
\usepackage{multirow}
\usepackage{booktabs}
\newtheorem{theorem}{Theorem}

\usepackage{mathtools}

\begin{document}
%
\title{High Rate Communication over One-Bit Quantized Channels via Deep Learning and LDPC Codes}
%
%
%

\author{Eren Balevi and
        Jeffrey G. Andrews \\ 
\IEEEauthorblockA{
		Department of Electrical and Computer Engineering \\
    The University of Texas at Austin, TX 78712, USA\\
    Email: erenbalevi@utexas.edu, jandrews@ece.utexas.edu}	
}

\maketitle

\begin{abstract}
This paper proposes a method for designing error correction codes by combining a known coding scheme with an autoencoder. Specifically, we integrate an LDPC code with a trained autoencoder to develop an error correction code for intractable nonlinear channels. The LDPC encoder shrinks the input space of the autoencoder, which enables the autoencoder to learn more easily. The proposed error correction code shows promising results for one-bit quantization, a challenging case of a nonlinear channel. Specifically, our design gives a waterfall slope bit error rate even with high order modulation formats such as 16-QAM and 64-QAM despite one-bit quantization. This gain is theoretically grounded by proving that the trained autoencoder provides approximately Gaussian distributed data to the LDPC decoder even though the received signal has non-Gaussian statistics due to the one-bit quantization. 
\end{abstract}


\section{Introduction}
Popular error correction codes including turbo codes, low density parity check (LDPC) codes and polar codes produce error rates very close to the Shannon limit for linear additive white Gaussian noise (AWGN) channels. However, these codes often do not perform well under hardware constraints and/or impediments. Two key trends in wireless communication systems are the march towards higher carrier frequencies and a very large number of antennas, both of which will introduce increasingly non-ideal hardware constraints such as low-resolution quantization and other nonlinear distortions.  Hence, designing an error correction code that is robust to such distortions is important for future communication systems, even if such distortions are analytically intractable.

In this paper, we propose a methodology to develop efficient error correction codes for one-bit quantization. The main reason behind the selection of this nonlinear channel model is associated with the fact that one-bit quantization brings severe nonlinearities. Thus, a channel code that can cope with one-bit quantization should be able to resist less severe nonlinearities. Furthermore, it is easy to model the nonlinearities due to one-bit quantization, which can be done by taking the sign of the real and imaginary part of the signals, unlike other RF nonlinearities. Our methodology relies on combining a state-of-the-art code that is optimized for a static linear AWGN channel (which can be seen as an outer code) with an autoencoder (which can be seen as an inner code) so as to capture the system dynamics via learning/training. This can also be interpreted as combining the current knowledge in coding theory with the recent advances in deep learning in an attempt to end up with novel error correction codes. 

Employing a neural network for error correction codes dates back to late eighties. More precisely, \cite{BruckBlaum89} shows how to decode linear block codes. Similarly, the Viterbi decoder was implemented with a neural network for convolutional codes in the late nineties \cite{WangWicker96}, \cite{HamalainenHenriksson99}. A simple classifier is learned in these studies instead of a decoding algorithm. This leads to a training dataset that must include all codewords, which makes them infeasible for most codes due to the exponential complexity. Recently, it was shown that a decoding algorithm could be learned for structured codes \cite{Gruber17}, however this design still requires a dataset with at least $90\%$ percent of the codebook, which limits its practicality to small block lengths. To learn decoding for large block lengths, \cite{KimViswanath18} trained a recurrent neural network for small block lengths that can generalize well for large block lengths. Although there are many papers that propose a deep learning-based decoding algorithm, there are only a few papers that aim to learn an encoder \cite{JiangViswanath18}, \cite{Kosaian18}. 

In this paper, we design an error correction code, i.e., learn an encoder-decoder pair for a severe nonlinear channel model: a one-bit quantized AWGN channel. For this purpose, we train an autoencoder, and then incorporate an LDPC code to this autoencoder.  In the case of QPSK or BPSK modulation one-bit quantization corresponds to hard decision decoding and only leads to a few dB signal-to-noise-ratio (SNR) loss. However, for high order modulations one-bit quantization leads to a very poor error rate, since the real and the imaginary part of the signals carry more than one bit information. The closest paper to this work is \cite{BalAndDeepECC}, which proposed to integrate a turbo code to an autoencoder to handle the detrimental effects of one-bit quantization for QPSK and 16-QAM signaling. The main difference of this paper is to (i) use an LDPC code instead of a turbo code; (ii) utilize a modified autoencoder architecture; and (iii) propose a simpler, but efficient training policy that gives us better error rate for high order modulations, e.g., 64-QAM can be operable for one-bit quantization at sufficiently low SNRs. 

The main contributions of this paper are as follows. First, we propose sending symbols at a faster-than-Nyquist rate to have a large enough learning capacity for the autoencoder that is integrated with an LDPC code. Then, we theoretically show that with this transmission and the proposed autoencoder architecture (hat has infinite width neural layers) it is possible to provide Gaussian distributed data to the LDPC decoder despite the nonlinearities of one-bit quantization. In what follows, we evaluate the efficiency of our error correction code with simulations for  practical finite layer widths. The proposed channel code can make high order modulations such as 16-QAM and 64-QAM operable even if one-bit quantization is employed. To be more precise, our code can approach the performance of the LDPC codes in 802.11n and DVB-S2 standards that decode unquantized samples without increasing bandwidth. This apparently brings significant spectral efficiency gain for low-resolution receivers by enabling high order modulations.



\section{Autoencoder Empowered Error Correction Coding}
\label{Autoencoder Empowered Error Correction Coding}
Conventional error correction codes do not provide any guarantee to obtain almost zero bit error rate at sufficiently low SNRs for nonlinear channels or non-Gaussian noise. To have a low bit error rate for more challenging nonlinear environments, we propose to integrate an autoencoder with an existing error correction code. This results in a concatenated code such that the outer code is a known state-of-the-art code and the inner code is a learned or trained autoencoder. This idea is illustrated for an LDPC code in Fig. \ref{fig:proposed_method}.
\begin{figure*} [!t]
\centering
\includegraphics[width=6in]{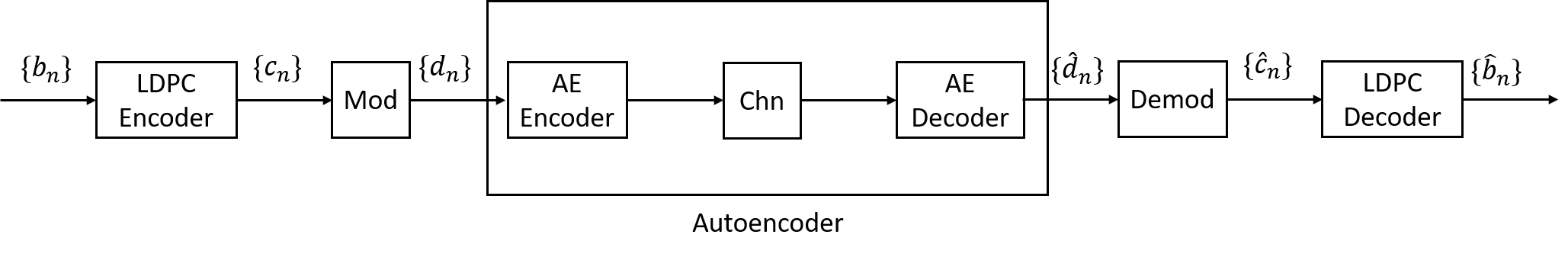}
\caption{The proposed error correction code that integrates an LDPC code with an autoencoder.}
\label{fig:proposed_method}
\end{figure*}
In accordance with that, the input bits $b_n$ are first encoded with an LDPC encoder. This LDPC encoder produces a codeword $c_n$ that has $B$ bits. This means that the LDPC encoder output block size is $B$. Then, these coded bits are modulated with an $M$-ary modulation method giving $B/\log_2(M)$ symbols per codeword, which are represented as $d_n$. The coded symbols $d_n$ are then given to the autoencoder.

Processing the entire symbols of a codeword with an autoencoder yields excessive computational complexity for large codeword lengths. For instance, taking a reasonable codeword length of $1000$ as an input to an autoencoder yields millions of parameters. Therefore, these coded symbols are broken into small blocks of $N$, in which $N\ll B/\log_2(M)$. More precisely, a codeword is written in terms of the modulated symbols as
\begin{equation} \label{coded_sym}
    \mathbf{d} = [\mathbf{d}_1^T\ \mathbf{d}_2^T \cdots \ \mathbf{d}_{S}^T]^T
\end{equation}
where $S=B/(N\log_2(M))$ and
\begin{equation}\label{coded_symbol}
    \mathbf{d}_{i} = [d_{iN-N+1}\ d_{iN-N+2}\ \cdots \ d_{iN}]^T.
\end{equation}
For \eqref{coded_sym} each $\mathbf{d}_i$ is processed separately, but with the same weights or parameters of the autoencoder. This means that a single set of weights (or a single autoencoder) is learned for all $\mathbf{d}_{i}$ for $i=1,2,\cdots,S$. This brings huge complexity savings.

There are many alternatives for the autoencoder architecture. For simplicity, we utilize a few fully connected layers as was done in \cite{BalAndDeepECC}, \cite{BalAnd19}. Since the same weights are used for each block, the overall architecture can be considered as a one-dimensional convolutional neural network with stride $N$. This autoencoder architecture is given in Table \ref{tab:AE}. In addition to the fully connected layers, there is a lambda layer meaning it does not have trainable weights. This lambda layer includes the transmission power normalization, physical channel, pulse shape, sampling and quantization.  
\begin{table} [!t] 
\renewcommand{\arraystretch}{1.3}
\caption{The autoencoder architecture}
\label{tab:AE}
\centering
\begin{tabular}{c|c|c|c|c}
    \hline
       Layer & Type & Size & Activation & Weights\\
    \hline
    \hline
$l_0$: Input & Coded Symbols & N & - & $-$\\
    \hline
$l_1$: Hidden Layer-1 & Fully Connected & GN & Linear & $\mathbf{\Theta}_{1}$\\
    \hline
$l_2$: Lambda layer & Channel & GN & - & - \\
    \hline
$l_3$: Hidden Layer-2 & Fully Connected & KN & ReLU & $\mathbf{\Theta}_{2}$\\
    \hline
$l_4$: Hidden Layer-3 & Fully Connected & KN & ReLU & $\mathbf{\Theta}_{3}$\\
    \hline
$l_5$: Hidden Layer-4 & Fully Connected & KN & ReLU & $\mathbf{\Theta}_{4}$\\
    \hline
$l_6$: Output & Fully Connected & N & Linear & $\mathbf{\Theta}_{5}$\\
		\hline
\end{tabular}
\end{table}

The autoencoder takes the coded symbols as blocks with length $N$ and generates
\begin{equation}
    \mathbf{e}_i = \mathbf{\phi_1}(\mathbf{\Theta}_{1}\mathbf{d}_i+\mathbf{b}_1) = \mathbf{\Theta}_{1}\mathbf{d}_i+\mathbf{b}_1
\end{equation}
where $\phi_1$ is an identity function, because there is a linear activation function for the first hidden layer as can be seen in Table \ref{tab:AE}. Furthermore, $\mathbf{\Theta}_{1}$ and $\mathbf{b}_1$ correspond to the trainable weights (in matrix form) and the biases (in vector form) of the first hidden layer. Both these weights and biases are initialized with Gaussian random variables that have zero-mean, and $\sigma_{\theta}^2$ and $\sigma_{b}^2$ variance, respectively as is standard practice \cite{GlorotBengio10}, \cite{HeSun15}. 
Notice that
\begin{equation}\label{fir_lay}
    \mathbf{e}_i = [e_{i1}\ e_{i2}\ \cdots \ e_{iGN}]^T,
\end{equation}
where $GN$ is the width of the first hidden layer.

The main idea to leverage an autoencoder for a coding scheme is to tackle all kinds of channel and hardware impediments so as to perfectly (or with very small error probability) transfer the coded symbols from a transmitter to a receiver. For this purpose, the autoencoder has to have a large capacity. This is the reason for encoding each coded symbol with $G$ neurons. This obviously decreases the bandwidth efficiency $G$-fold if conventional orthogonal transmission methods are employed. A  better approach is to send the input symbols faster, which is known as non-orthogonal faster-than-Nyquist transmission, at the expense of creating inter-symbol interference (ISI) and colored noise \cite{AndersonOwall}. This method enables us to use a large value for $G$ without any increase in bandwidth at the expense of degrading the minimum distance between the encoded LDPC symbols, since an autoencoder does not have isometry property, i.e., it does not preserve the distances among different inputs. 
Furthermore, in traditional communication systems faster-than-Nyquist signaling heavily increases the demodulation complexity for a large $G$. However, this is not an issue when it comes to an autoencoder, because a (neural) decoder with the same complexity is used irrespective of how symbols are transmitted.

The symbols in \eqref{fir_lay} are transmitted as
\begin{equation}
    s(t) = \sqrt{\rho}\sum_{i=1}^{S}\sum_{n=1}^{GN}e_{in}h\left(t-(i-1)NT-\frac{nT}{G}\right)
\end{equation}
where $\rho$ is the transmission power, $h(t)$ is the real pulse shape and $T$ is the symbol period for orthogonal transmission. The received continuous time signal over an AWGN channel is filtered with a matched filter $h(-t)$ to yield
\begin{equation}\label{rec_sig}
    y(t) = (s(t) + z(t)) * h(-t),
\end{equation}
where $z(t)$ is a zero-mean Gaussian noise with variance $\sigma_z^2$ and $*$ denotes linear convolution. Writing \eqref{rec_sig} in integral form gives us
\begin{equation} \label{rec_sig_int}
    y(t) = \int_{-\infty}^{\infty}(s(\tau)+z(\tau))h(\tau-t)d\tau
\end{equation}
Sampling \eqref{rec_sig_int} at $t=(i-1)NT+\frac{kT}{G}$ with a sampling period of $\frac{kT}{G}$ makes our received symbols
\begin{equation}\label{sampled_rec_sig}
    y_{ik} = \sqrt{\rho}\sum_{i=1}^{S}\sum_{n=1}^{GN}e_{in}g[k-n]+z_{ik}
\end{equation}
where
\begin{eqnarray}\nonumber 
    g[k-n]&=&\int_{-\infty}^{\infty}h\left(\tau-(i-1)NT-\frac{nT}{G}\right) \times \\ \nonumber
     && h\left(\tau-(i-1)NT-\frac{kT}{G}\right)d\tau
\end{eqnarray}
and
\begin{equation}\nonumber
    z_{ik} =\int_{-\infty}^{\infty}z(\tau)h\left(\tau-(i-1)NT-\frac{kT}{G}\right)d\tau.
\end{equation}

Sending symbols faster than the symbol period results in inter-symbol interference and colored noise. To make this clearer, we consider the vector-matrix representation of the received signal 
\begin{equation}\label{unquant_samp}
    \mathbf{y}_i = [y_{i1}\ y_{i2}\ \cdots \ y_{iGN}]^T,
\end{equation}
which can be expressed as
\begin{equation}\label{ISI_Model}
    \mathbf{y}_i = \mathbf{G}_{\rm ISI}\mathbf{e}_i + \mathbf{z}_i
\end{equation}
where $\mathbf{G}_{\rm ISI}$ is a $GN \times GN$ Toeplitz matrix, and its first row becomes $[g[0]\ g[1]\ \cdots\ g[GN-1]]$. Notice that $\mathbf{G}_{\rm ISI}$ would be an identity matrix if orthogonal transmission was used. The correlation of the noise samples is
\begin{equation}\label{corr_noise}
    \mathbb{E}[z_{ik}z_{in}] = \sigma_z^2g[k-n].
\end{equation}

The sampled signal in \eqref{sampled_rec_sig} is quantized before further processing as 
\begin{equation}\label{quan_layer}
    r_{ik} = \mathcal{Q}(y_{ik}). 
\end{equation}
In what follows, the decoder part of the autoencoder takes a block of $GN$ samples of 
\begin{equation} \label{proc_sig}
    \mathbf{r}_i = [r_{i1}\ r_{i2}\ \cdots \ r_{iGN}]^T.
\end{equation}
The signal in  \eqref{proc_sig} is decoded with the second, third and fourth hidden layers, which have a width of $KN$ and ReLU activation function, and the output layer, which has a width of $N$ and a linear activation function. This gives us the estimate of \eqref{coded_symbol}. Mathematically,
\begin{equation}\label{rec_coded_symbol}
    \mathbf{\hat{d}_i} = \mathbf{\Theta}_{5}\mathbf{\phi_4}(\mathbf{\Theta}_{4}\mathbf{\phi_3}(\mathbf{\Theta}_{3}\mathbf{\phi_2}(\mathbf{\Theta}_{2}\mathbf{r}_i+\mathbf{b}_2)+\mathbf{b}_3)+\mathbf{b}_4)+\mathbf{b}_5.
\end{equation}
The parameters $\mathbf{\Theta}_{1},\mathbf{\Theta}_{2},\mathbf{\Theta}_{3},\mathbf{\Theta}_{4},\mathbf{\Theta}_{5}$ and the biases $\mathbf{b}_1, \mathbf{b}_2, \mathbf{b}_3, \mathbf{b}_4, \mathbf{b}_5$ are denoted as $\mathbf{W}$ for brevity. These are optimized according to a squared error loss function as
\begin{equation} \label{error_func}
    \mathbf{W}^* = 
 \underset{\mathbf{W}}{\text{arg\ min}} ||\mathbf{d_i}-\mathbf{\hat{d}_i}||_2^2  \\ 
\end{equation}
where $\mathbf{d_i}$ and $\mathbf{\hat{d}_i}$ are defined in \eqref{coded_symbol} and \eqref{rec_coded_symbol}, respectively. Gathering all these blocks constitutes the estimate of one transmitted codeword as
\begin{equation}
    \mathbf{\hat{d}} = [\mathbf{\hat{d}_1}^T\ \mathbf{\hat{d}_2}^T \cdots \ \mathbf{\hat{d}_S}^T]^T.
\end{equation}

\section{Theoretical Guarantees}
\label{Theoretical Guarantees}
The output of the autoencoder can be written according to its input as
\begin{equation} \label{ae_out}
\mathbf{\hat{d}_i} = \mathbf{d_i} + \mathbf{v_i}
\end{equation}
where $\mathbf{v_i}$ refers to the residual error stemming from the noise and quantization. Since \eqref{ae_out} is given to the LDPC decoder, which is optimized according to the data that has Gaussian statistics, it is important to determine the distribution of $\mathbf{v_i}$. Next we prove that $\mathbf{v_i}$ has a Gaussian distribution.

\begin{theorem}
The autoencoder provides Gaussian distributed data to the LDPC decoder, i.e., $\mathbf{v_i}$ has a Gaussian distribution when the autoencoder has infinitely large width neural layers and is trained with gradient descent for faster-than-Nyquist transmission and one-bit quantization. 
\end{theorem}

\begin{proof}
The autoencoder architecture in Table \ref{tab:AE} can be expressed layer-by-layer as
\begin{equation} 
\begin{split}
l_0: & z^{(0)} = \mathbf{d}_i, x^{(0)} = z^{(0)}  \\  
l_1: & z^{(1)} = \mathbf{\Theta}_{1}x^{(0)} +\mathbf{b}_1 , x^{(1)} = {\phi_1}(z^{(1)} ) \\  
l_2: & z^{(2)} =  \mathbf{y}_i = f(x^{(1)}), x^{(2)}  = \phi_2(z^{(2)}) \\   
l_3: & z^{(3)} =\mathbf{\Theta_{2}}x^{(2)} + \mathbf{b_{2}}, x^{(3)}  = \phi_3(z^{(3)}) \\  
l_4: & z^{(4)} = \mathbf{\Theta_{3}}x^{(3)} + \mathbf{b_{3}}, x^{(4)}  = \phi_4(z^{(4)}) \\  
l_5: & z^{(5)} = \mathbf{\Theta_{4}}x^{(4)} + \mathbf{b_{4}}, x^{(5)}  = \phi_5(z^{(5)}) \\  
l_6: & z^{(6)} = \mathbf{\Theta_{5}}x^{(5)} + \mathbf{b_{5}} \\
\end{split}
\end{equation}
where $f(\cdot)$ corresponds to the lambda layer except quantization and $ \mathbf{y}_i$ is defined in \eqref{unquant_samp}. Also, $\phi_{2}(\cdot) = \mathcal{Q}(\cdot)$. Since all the weights and biases are initialized with Gaussian random variables, for each unit (or neuron) in the $l^{th}$ layer $z_i^{(l)}|x^{(l-1)}$ is an identical and independent Gaussian random variable with zero mean and covariance 
\begin{equation} \label{cov_l}
K^{(l)}(z,\hat{z}) = \sigma_b^2 + \sigma_\theta^2 \mathbb{E}_{z_i^{(l-1)}} [\phi_{l-1}(z_i^{(l-1)})\phi_{l-1}(\hat{z}_i^{(l-1)})]
\end{equation}
except for $l=2$. For the second layer $z_i^{(2)}|x^{(1)}$ is also a Gaussian random variable due to \eqref{ISI_Model} and Gaussian noise. Specially, $z_i^{(2)}|x^{(1)}$ has zero-mean and its covariance is as given by \eqref{corr_noise}.

As the width goes to infinity, \eqref{cov_l} can be written in integral form as given in \eqref{cov_int}. 
\begin{figure*}[!h]
\small
\begin{equation}\label{cov_int}
    \underset{ {N\rightarrow   \infty}}{\text{lim}}K^{(l)}(z,\hat{z}) = \int\int \phi_{l-1}(z_i^{(l-1)})\phi_{l-1}(\hat{z}_i^{(l-1)}) 
\mathcal{N}\left(z,\hat{z};0,\alpha_\theta^2\left[
  \begin{array}{cc}
  K^{(l-1)}(z,z) & K^{(l-1)}(z,\hat{z})  \\
  K^{(l-1)}(\hat{z},z) & K^{(l-1)}(\hat{z},\hat{z})   
  \end{array}
\right] + \alpha_b^2 \right)dzd\hat{z}.\\
\end{equation}
\end{figure*}
\begin{figure*}[!t]
\centering
\subfigure[16-QAM]{
\label{fig:648_16qam}
\includegraphics[width=3.25in]{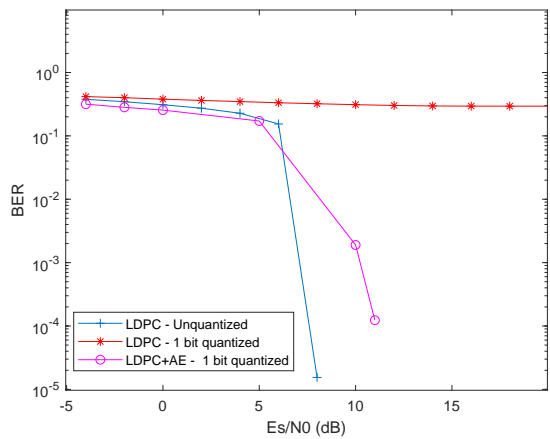}}
\qquad
\subfigure[64-QAM]{
\label{fig:648_64qam}
\includegraphics[width=3.25in]{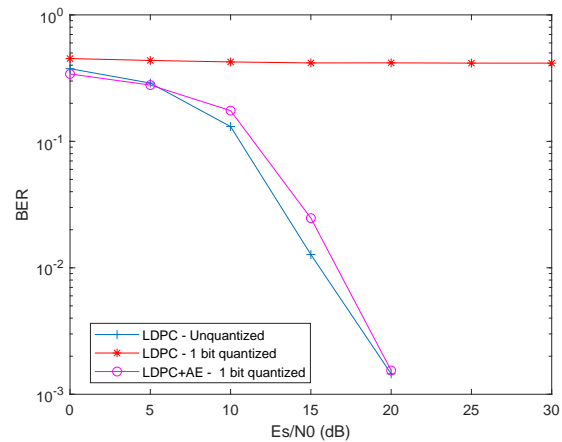}}
\caption{The error rate for the code rate $\frac{1}{2}$ LDPC code with a codeword length of 648 bits}
\end{figure*}
To be more compact, the double integral in \eqref{cov_int} can be represented with a function such that 
\begin{equation}
\underset{ {N\rightarrow \infty}}{\text{lim}}K^{(l)}(z,\hat{z}) = F_{l-1}(K^{(l-1)}(z,\hat{z})).
\end{equation}
Hence, $z^{(6)}|z^{(0)}$ is a Gaussian process with zero mean and covariance 
\begin{equation}
K^{(6)}(z,\hat{z}) = F_5(F_4(F_3(F_2(F_1(F_0(K^{(0)}(z,\hat{z})))))))
\end{equation}
when $N \rightarrow \infty$. This means that the output of the autoencoder yields Gaussian distributed data in the initialization phase.

During training, the parameters are iteratively updated at time $n$ as
\begin{equation}
\mathbf{W}^n = \mathbf{W}^{n-1}-\eta \nabla_{\mathbf{W}^{n-1}} L(\mathbf{W}^{n-1})
\end{equation}
where $\mathbf{W}^n = \{\mathbf{\Theta^n_{1}}, \cdots, \mathbf{\Theta^n_{5}}, \mathbf{b^n_1}, \cdots, \mathbf{b^n_5} \}$, and $L(\cdot)$ is the loss function. In parallel, the output $z^{(6)}$ is updated as 
\begin{equation} \label{out_upd}
z^{(6),n} = z^{(6),{n-1}} + \nabla_{\mathbf{W}^{n-1}} (z^{(6),{n-1}})(\mathbf{W}^n - \mathbf{W}^{n-1}).
\end{equation}
The gradient term in \eqref{out_upd} is a nonlinear function of the parameters. Nevertheless, it was recently proven in \cite{JaehoonLee19} that as the width goes to infinity, this nonlinear term can be linearized via a first-order Taylor expansion. More precisely,
\begin{equation} \label{out_ae}
z^{(6),n}  =  z^{(6),0} + \nabla_{\mathbf{W}_{0}}(z^{(6),0})(\mathbf{W}^{n}-\mathbf{W}^0) + \mathcal{O}(N^{-0.5} )
\end{equation}
where the output at the initialization or $z^{(6),0}$ is Gaussian as discussed above. Since the gradient (and hence the Jacobian matrix) is a linear operator, and a linear operation on a Gaussian process results in a Gaussian process, the output of the autoencoder for a given input (or $z^{(6),n}|z^{(0),n}$) is a Gaussian process throughout training with gradient descent. It is worth emphasizing that having a piece-wise quantization function in the second layer does not violate the aproximation in \eqref{out_ae}, because $\mathcal{Q}(\cdot)$ can easily be approximated to a sigmoid function.
\end{proof}


\section{Simulations}
\label{Simulations}
We simulate the performance of the proposed coding method by combining the autoencoder with (i) the LDPC code that has a code rate of $\frac{1}{2}$ and a codeword length of 648 bits and (ii) the LDPC code that has a code rate of $\frac{1}{2}$ and a codeword length of 64800 bits. These LDPC codes have been used in 802.11n and DVB-S2 standards, respectively. Our performance metric is the bit error rate (BER) with respect to the energy per symbol. Throughout the simulation, the symbols are sent $10$ times faster than Nyquist rate, which yields a strong ISI. 

The 6-layer autoencoder architecture is trained with the squared error function given in \eqref{error_func} with $G=10$ and $K=20$. This leads to leaving the layer before quantization (or $l_1$) untrained. Using an autoencoder whose input layer is much smaller than the number of coded LDPC symbols per codeword makes training further challenging. To handle these issues, we propose to periodically train the architecture for $k$ codewords and then utilize it for these $k$ codewords. The main reason for this training policy is related with the very poor generalization capability of the neural network due to one-bit quantization.

For the first (shorter codeword) LDPC code, the coded bits are first modulated with 16-QAM and then fed into the autoencoder. These coded symbols are processed by the autoencoder in blocks of $24$, i.e., $N=24$.  As can be observed in  Fig. \ref{fig:648_16qam}, using an LDPC code alone is not sufficient for one-bit quantization despite the fact that it decays very rapidly in the case of unquantized samples after 5dB. On the other hand, integrating an autoencoder with this LDPC code brings substantial improvement and leads to obtain a close performance with respect to the unquantized LDPC code. 
We repeat this experiment for 64-QAM in Fig. \ref{fig:648_64qam}. In comparison to 16-QAM modulation, our coding method that only sees one-bit quantized samples gives nearly the same performance with the LDPC code that processes the unquantized samples. Similar to 16-QAM modulation, this LDPC code alone does not work properly for 64-QAM if there is a one-bit ADC in the receiver.

The performance of the proposed code is also assessed for larger codeword lengths by integrating the autoencoder to the LDPC code that has a codeword length of 64800 bits. All the hyper-parameters of the autoencoder remain the same except that $N$ is taken as $64$ instead of $24$. This is associated with the fact that large blocks are needed for the autoencoder to capture the structures for large codewords. We observe nearly the same behavior as compared to the shorter LDPC codes as depicted in Fig. \ref{fig:64800_16qam} for 16-QAM signaling. 
\begin{figure} [!h]
\centering
\includegraphics[width=3.25in]{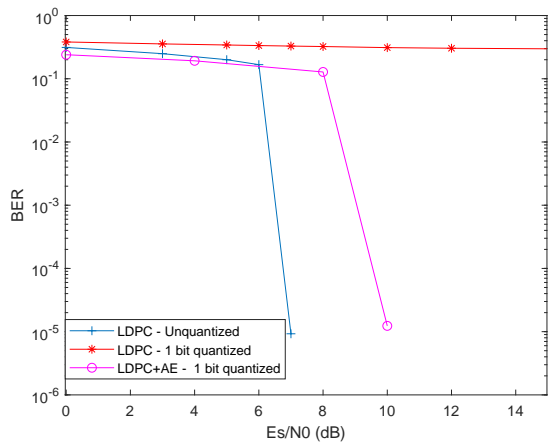}
\caption{The error rate for the code rate $\frac{1}{2}$ LDPC code with a codeword length of 64800 bits for 16-QAM}
\label{fig:64800_16qam}
\end{figure}

\section{Conclusions}
\label{Conclusions}
In this paper, a new design methodology is discussed for developing error correction codes for nonlinear channels by leveraging the merits of deep learning and using the current knowledge in coding theory. This idea is utilized to design a channel code for one-bit quantization. Our results show that the proposed method makes higher-order modulation formats operable for one-bit receivers. This obviously brings in a large spectral efficiency gain. As future work, better autoencoder architectures can be designed instead of using a couple of fully connected layers so as to improve our results. Furthermore, it is interesting to craft novel loss functions for one-bit quantization instead of using the canonical squared loss function, which heavily affects the training policy.

\bibliographystyle{IEEEtran}

\end{document}